\documentclass[letterpaper, 10 pt, conference]{ieeeconf}  

\IEEEoverridecommandlockouts                              
\usepackage{graphics} 
\usepackage{amsmath} 
\usepackage{amssymb}  

\usepackage{graphicx}

\usepackage{amsthm}
\usepackage{color}

\newcommand{\esh}{\widehat e}
\newcommand{\Jes}{J(e,\esh)}
\newcommand{\Jesh}{J^{\star}(\esh)}
\newcommand{\ps}{p^{\star}(e,\esh)}

\newtheorem*{problem*}{Problem Statement}

\newcommand{\mc}{\mathcal}
\newcommand{\mb}{\mathbb}
\newcommand{\mf}{\mathbf}

\newcommand{\T}{\intercal}

\newtheorem{proposition}{Proposition}
\newtheorem{lemma}{Lemma}
\newtheorem{theorem}{Theorem}

\title{\LARGE \bf
Incentivizing Truthful Reporting from Strategic Sensors \\
in Dynamical Systems
}

\author{Yuelin Zhao and Roy Dong
\thanks{Y. Zhao (corresponding author) is with the Department of Mechanical Science and Engineering, University of Illinois at Urbana-Champaign, Urbana, IL 61801, USA
        {\tt\small yuelinz3@illinois.edu}.}%
\thanks{R. Dong is with the Department of Electrical and Computer Engineering, University of Illinois at Urbana-Champaign, Urbana, IL, 61801, USA
        {\tt\small roydong@illinois.edu}.}%
}

\begin{document}

\maketitle
\thispagestyle{empty}
\pagestyle{empty}

\begin{abstract}
Human agents are increasingly serving as data sources in the context of dynamical systems. Unlike traditional sensors, humans may manipulate or omit data for selfish reasons. Therefore, this paper studies the influence of effort-averse strategic sensors on discrete-time LTI systems. In our setting, sensors exert costly effort to collect data, and report their effort to the system operator. However, sensors do not directly benefit from the output of the system, so they will not exert much effort to ensure accuracy and may even falsify their reported effort to maximize their utility. We explore payment mechanisms that incentivize truthful reporting from strategic sensors. We demonstrate the influence of the true and reported effort on the expected operational cost. Then, we use the realizations of the system cost to construct a payment function. We show that payment functions typically used in static settings will not be able to elicit truthful reports in general, and present a modified payment function that elicits truthful reporting, which requires terms that compensate for the dynamic impact of reported efforts on the closed-loop performance of the system.
\end{abstract}

\section{Introduction}
\label{sec:introduction}


Human agents are increasingly serving as data sources in the context of dynamical systems. However, unlike conventional sensors, human agents can manipulate or even falsify data to achiever their personal goals. The strategic behavior of data sources can have a significant impact on the control and operation of the system. 

For example, navigation apps like Google Maps and Waze combine data from real-world sensors and user data to estimate the traffic flow and congestion on road networks. However, when interacting with these apps, users are more interested in minimizing their own travel times rather than the overall traffic flows, so they may not put effort in reporting their data, or even falsify their data to achieve their personal goals. Such behavior can lead to worse traffic conditions. In this paper, we explore the design of incentives that help align the goals of strategic data sources and system operators.

In reality, data sources often incur an effort cost when obtaining high quality data. In this paper, we model our strategic sensors as effort-averse, i.e. all else equal, they will prefer to exert less effort and share lower quality data. 
Furthermore, system operators do not have access to data sources' private information. In particular, the system operators do not know the effort exerted by the data sources, nor the exact distribution from which their data is drawn. We explore how incentives can compensate for the problem of moral hazard that arises due to this asymmetric information~\cite{Nisan_game_theory_2007}.

The contribution of this paper is to study the influence of a single strategic sensor on the operational cost of a dynamical system. We decompose the operational cost into terms that depend on reported information and terms that depend on the hidden effort level. Given this decomposition, we show that payment mechanisms designed for static situations do not incentivize truthful reporting in our setting, and we provide a modified payment method that ensures truthful reporting for parameters that satisfy certain conditions.

In particular, we consider a discrete-time, linear, time-invariant system minimizing a quadratic cost on a finite time horizon. In our setting, strategic sensors report an effort level $\esh$ and exert an effort $e$, which may not necessarily equal $\esh$. We demonstrate the influence of $\esh$ and $e$ on the system cost and investigate payment methods that enforce truthful reporting, i.e. ensures that it is in the strategic sensor's best interest to choose $\esh = e$. 

The rest of this paper is organized as follows. Section~\ref{sec:background} shows the related literature and contextualizes our work. We present our model and formally outline the interaction between the operator and strategic sensor in Section~\ref{sec:model}. We define the control policy in Section~\ref{sec:sys_op}. In Section~\ref{sec:cost}, we analyze the system cost and expected system cost as a function of the true effort $e$ and reported effort $\esh$, providing a decomposition of the expected system cost $\mb{E}[J]$ into terms that only depend on $e$ and $\esh$. We use this decomposition to construct payment functions in Section~\ref{sec:payment}. Finally, we close with final remarks and present some avenues for future work in Section~\ref{sec:conclusion}.

\section{Background}
\label{sec:background}

Strategic data sources and the corresponding incentive mechanism design problem have received great attention in recent years. Most of literature on them fall into two categories, based on whether the application domain is static or dynamic.

Some of the existing work studies how to incentivize strategic data sources in a static setting, where there is no underlying dynamical system, and the data collectors are typically trying to estimate some underlying fixed function.
In~\cite{D.G_2016}, an optimal contract is presented that minimizes the total payment of the estimator while guaranteeing strategic sensors to put in sufficient effort and truthfully report the estimate, by assuming one of the sensors is loyal, i.e. reporting true information. \cite{Chen_2019} designs an optimal mechanism that minimizes the expected total compensation to the strategic sources while guaranteeing certain level of estimation accuracy. \cite{T.W_2019} models a data market with multiple data aggregators and multiple data sources, and shows that such coupling would lead to either infinite many equilibria or none, where all equilibria can be socially inefficient.~\cite{cai2014optimum} proposes an optimal mechanism of statistical estimators that minimizes the weighted sum of payments and estimation error, and shows that this mechanism is extremely robust, i.e, each data source's decision is a unique dominant strategy.~\cite{D.G_2017} studies the estimation problem in a repeated setting and designs a compensation scheme that employs stochastic data verification and builds a reputation history for each sensor. 

More recently, some literature has explored the impact of strategic sensors and payment design in the context of stochastic dynamical systems.~\cite{Teixeira_2015} studies strategic stealthy false-data injection attacks problem on discrete time linear systems. The author provides the necessary and sufficient conditions for the sensitivity metric to be unbounded, and further proposed a novel attack policy in~\cite{Teixeira_2019}. ~\cite{Pedro_2019} designs a mechanism for dynamical systems that ensures optimal control, and truthfully reporting forms a Nash equilibrium among the strategic agents. ~\cite{KeMa_2019} introduces an extension of Vickrey-Clarke-Groves (VCG) mechanism that guarantees Incentive Compatibility and truthful reporting of strategic agents in a linear-quadratic-Gaussian (LQG) dynamical systems. By carefully constructing a sequence of layered VCG payments, such mechanism can also ensure that the mechanism is budget-balanced and satisfies individual rationality under certain conditions. 

Most of this previous work assumes some sort of adversarial intent or some stake in the system operation; in contrast, our work focuses primarily on an effort-averse data source. Additionally, most of these works require multiple agents to enforce truthfulness amongst each other; in contrast, we consider the case of a single strategic sensor, and use the dynamics of the system itself to enforce truthfulness.

The terms and setting in this paper are closely related to~\cite{T.W_2019}. Our contribution is to introduce the effort-averse behavior of strategic data sources in a dynamic setting. In our problem, strategic sensors can exert a costly effort $e$ to reduce their measurement noise. Because $e$ is a private information to strategic sensors, we allow sensors to report their effort level as $\esh$ before they exert $e$. That is, the system uses the control law designed with $\esh$ and incurs the true measurement noise with $e$. However, strategic sensors do not directly benefit from the system performance, and they can even falsity $\esh$ to maximize their profit. Therefore, an appropriate payment method is desired to ensure truthfully reporting from sensors. 

Given that the strategic sensors truthfully report, the problem reduces to a classical control problem: minimizing the expected quadratic cost of an LQG system.
To our best of our knowledge, this is one of the first papers studying the impact of effort-averse strategic sensors on the system costs in the LQG setting. The reported noise covariances are used to design our controllers, and we separate out the influence of the true noise covariance  and the reported noise covariances on the closed-loop performance of our system. Using information available to a system operator, we design payments that incentivize truthful reporting of the noise covariances. 


%


\begin{table*}[h]
    \centering
    \begin{tabular}{|c|l|c|}
        \hline
        \textbf{Notation} & \textbf{Meaning} & \textbf{Defined or First Used in
        Equation} \\\hline
        $A, B, C, C^r, C^s$ 			& System dynamics 					& \eqref{eq:linear_system} \\
        $w_k, v_k, v_k^r, v_k^s$ 		& Process noise and observation noise 	& \eqref{eq:linear_system} \\
        $\sigma^2(e)$ 				& The effort-to-variance mapping 		& \eqref{eq:effort-covariance} \\
        $u(p,e)$ 					& The strategic sensor's utility function 	& \eqref{eq:sensor_util} \\
        $N$ 						& The time horizon of optimization 		& \eqref{eq:def_J} \\
        $J, Q, R$ 					& The system operator's cost function and cost function parameters	& \eqref{eq:def_J} \\
        $\Bar x, \Bar A, \Bar B$ 		& Augmented, closed-loop state variables and dynamics 			& \eqref{eq:bar_defs1} \\
        $\mathbf{\Bar A}, \mathbf{\Bar B}, \mathbf{\Bar V}$ & Iterated matrix applications 	& \eqref{eq:augment_system} \\
        $\mathbf{\bar{Q}_k}$ 		& Quadratic cost for the augmented system 	& \eqref{eq:bold_q} \\
        $\mf{\bar{\Sigma}}_k(e), \mf{\bar{\Sigma}}_{k,1},\mf{\bar{\Sigma}}_{k,2}$ 	& Covariance terms for the augmented noise 		& \eqref{eq:bold_sigk} \\
        $\mathbf{\Sigma}_1, \mathbf{\Sigma}_2$ 		& Initial augmented covariance terms 	& \eqref{eq:var_augment_system} \\
        $f_1(\esh), f_2(\esh)$ 		& The decomposition of $\mb{E}[\Jes]$ 			& \eqref{eq:f_i} \\
        $J^*(e)$ 					& The expected cost under truthful reporting 		& \eqref{eq:jstar} \\
        \hline
    \end{tabular}
    \caption{Notation Reference Chart}
    \label{tab:notation}
\end{table*}

\section{Model}
\label{sec:model}


In this section, we present our model. We look at a finite time horizon control problem for a linear, time-invariant (LTI) system in discrete time. At a high level, a system operator wishes to calculate a sequence of control inputs $(u_k)_k$ which minimizes their operational cost. Doing so, they have access to readings from two types of sensors. The regular sensor provides measurements with a known covariance $(y_k^r)_k$. The strategic sensor exerts some level of effort $e$ to provide measurements $(y_k^s)_k$; the level of effort $e$ affects the covariance of the measurements, but neither are known by the system operator. The system operator must issue incentives to the strategic sensor to ensure that the operating costs are low. 
The overall flow of information and control is summarized in Figure~\ref{fig:block_diagram}. For convenience, we've also compiled all the named variables in Table~\ref{tab:notation}.

\begin{figure}[!h]
\centering
\includegraphics[width=0.4\textwidth]{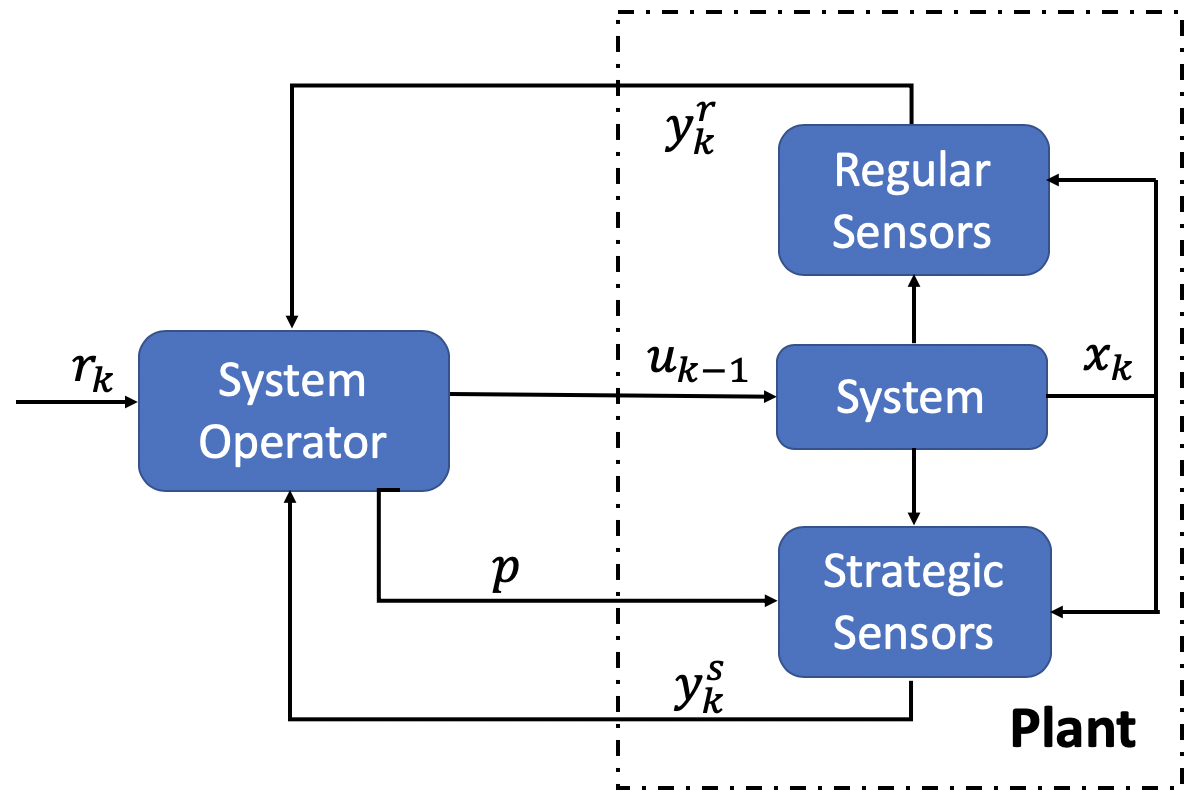}
\caption{Information and control flows in our model.}
\label{fig:block_diagram}
\end{figure}

The system dynamics are given by:
\begin{equation}
\label{eq:linear_system}
\begin{array}{l}{x_{k+1}=A x_{k}+B u_{k}+w_{k}} \\ {y_{k}=\left[\begin{array}{l}{y_{k}^{r}} \\ {y_{k}^{s}}\end{array}\right]= Cx_k + v_k = \left[\begin{array}{l}{C^{r}} \\ {C^{s}}\end{array}\right]x_k + \left[\begin{array}{l}{v_{k}^{r}} \\ {v_{k}^{s}}\end{array}\right] }\end{array} 
\end{equation}
Here, $w_{k}$, $v_{k}^{r}$, and $v_{k}^{s}$ are the process noise, and measurement noise for the regular and strategic sensors, respectively. 

We assume that all noise processes are independent, zero-mean Gaussians. Additionally, the covariance of the strategic sensor's measurements depend on a hidden level of effort $e$.
\[
\begin{aligned} &
x_{0} \sim \mc{N}(0,\Sigma_{x_{0}}) \qquad w_{k} \sim \mc{N}(0,\Sigma_{w}) \\ &
v_k = \left[\begin{array}{l}{v_{k}^{r}} \\ {v_{k}^{s}}\end{array}\right] \sim \mc{N}(0,\Sigma_{v}(e)), \quad \Sigma_{v}(e) = \begin{bmatrix} \Sigma_{v_{r}} & 0 \\ 0 & \Sigma_{v_{s}}(e)\end{bmatrix}
\end{aligned}
\]

For simplicity, in this paper, we assume the following form for the strategic sensor's covariance:
\begin{equation}
\label{eq:effort-covariance}
\begin{split}
& \Sigma_{v_{s}}(e) = \sigma^2(e) I
\end{split}
\end{equation}
Here, $\sigma^2(\cdot)$ is the mapping from effort $e$ to variance $\sigma^2(e)$. Throughout this paper, we will assume that $\sigma^2(\cdot)$ is strictly decreasing, convex, and twice continuously differentiable. Thus, as the level of effort increases, the variance decreases, and effort has diminishing returns. And $\sigma^2(\cdot)$ is a common knowledge for sensors and the system operator. 

The system operator announces a (potentially random) payment function $p$ to the strategic sensor. In Section~\ref{sec:payment}, we will provide a more detailed description of the payment function and which information it depends on. In particular, it is important that the payment $p$ depend on the effort level $e$, but the system operator does not have direct access to $e$.

Given this payment function $p$, the strategic sensor chooses $e$ to optimize its utility $u$ that is defined as:
\begin{equation}
\label{eq:sensor_util}
u(p,e) = \mb{E}[p] - e
\end{equation}
Equation~\eqref{eq:sensor_util} implies that the strategic sensor is effort-averse: all else equal, they would prefer to choose a smaller value of $e$. Additionally, it assumes that the strategic sensor must choose an effort level $e$ prior to the realization of the payment values and commit to it, i.e. that they must behave ex-ante. Furthermore, the strategic sensor is risk-neutral regarding these payments.

The system operator's random cost function is:
\begin{equation}
\label{eq:def_J}
J = 
\sum_{k = 0}^{N-1} \left( x_k^T Q x_k + u_k^T R u_k \right) + x_N^T Q x_N
\end{equation}
Here, $Q \succ 0$ and $R \succ 0$ are given positive definite matrices. 

However, the system operator does not have access to $e$, which is a private information for strategic sensors. Therefore, we allow strategic sensors to report their effort level $\esh$ before they exert $e$. That is, the system operator uses $\esh$ to design LQR and Kalman filter and run the system where the actual measurement with noise is determined by $e$.
Now, we present the order of play in this interaction, and outline which information is available to agents when they make decisions.

\begin{enumerate}
\item The system operator announces the payment function $p$. The system operator knows the function $\sigma^2(\cdot)$, but not the effort $e$ nor the variance $\sigma^2(e)$. 
\item The strategic sensor reports an effort level $\esh$, and exerts an effort level $e$. 
\item At each time $k$, the strategic sensor shares measurement $y_k^s$ with covariance $\Sigma_{v_s}(e)$, and the system operator makes a control decision $u_k$. The control decision can depend on $y_k^s$ and $\esh$, but not $\Sigma_{v_s}(e)$ or $e$. This repeats for $k = 0,\dots,N$.
\item The system operator issues the promised payment to the strategic sensor.
\end{enumerate}

Ideally, if the strategic sensor is truthful, i.e $\esh = e$, then we can minimize the expectation of system cost $\mb{E}[J]$ by  choices of inputs $(u_k)_k$ with corresponding controller and observer. However, strategic sensors can falsify their $\esh$ to maximize their utilities. Therefore, how would $\mb{E}[J]$ will be affected if the control policy is designed with the information that is far from the truth, i.e $\esh$ is far from $e$? And can we construct a payments $p$ to ensure that strategic sensors truthfully report the effort they exert?
This leads us to our problem statement.

\begin{problem*}

Given the system operator's available information, can they design a payment function $p$ that incentivizes truthful reporting (i.e. the strategic sensor is incentivized to report $\esh = e$)?

\end{problem*}

 As a first step, we assume the system operator merely wishes to minimize $\mb{E}[J]$, and can offer any $p$ without incurring any additional cost. Although an unreasonable assumption in practice, we believe the insights of this paper provide an interesting first step. Even under this assumption, the crux of the problem still remains: how can the system operator, given his limited information, incentivize the strategic data source appropriately? 

For example, setting $p$ that depends only on $\esh$ will not incentivize the strategic sensor to exert effort. Instead, it will make strategic sensor report a much higher $\esh$ to maximize their utility without increasing their true effort $e$ at all.

\section{Controller and Observer Design Under Truthful Reporting}
\label{sec:sys_op}


In this section, we outline how the system operator control inputs. We design a controller and observer assuming that the strategic sensor truthfully reports (i.e. that $\esh = e$). Then, in later sections, we will focus on designing payment contracts $p$ that induce truthful reporting.

First, supposing that the true measurement covariance is $\Sigma_v(\esh)$, we can invoke the separation principle to minimize $\mb{E}[J]$. That is, we can minimize our cost by estimating the state with a Kalman filter, and push that state estimate through a gain set by a linear-quadratic-regulator (LQR). Let $L_k(\esh)$ denote the Kalman filter gain, $K_k(\esh)$ denote the LQR feedback gain, and $P_k$ denote the LQR cost-to-go matrix. (See~\cite{Bertsekas:Optimal_control_2000} for more details.) Thus, the feedback control policy is:
\begin{equation}
\label{eq:LQG}
\begin{array}{rcl}
\hat{x}_{k+1} & = & A \hat{x}_{k}+B u_{k}+L_{k+1}\left(y_{k+1}-C\left(A \hat{x}_{k}+B u_{k}\right)\right) \\
u_{k} & = & K_{k} \hat{x}_k
\end{array} 
\end{equation}
Let $\hat{x}_k$ denote the state estimate at time $k$, and let $\Sigma_k$ and $\Sigma_{k+1 | k}$ denote the intermediate covariance estimates in the Kalman filter calculations.

Thus, given a reported effort level $\esh$, this defines a control policy given by the linear-quadratic-Gaussian (LQG) controller defined in Equation~\eqref{eq:LQG}. This allows us to think of the system operator's random cost $J$ as a function of two arguments: $J(e, \esh)$.

\section{Analysis of the Cost $J(e,\esh)$}
\label{sec:cost}


In this section, we outline properties of $J(e, \esh)$ and its expectation $\mb{E}[J(e,\esh)]$. Whereas the underlying true level of effort $e$ is unknown, the system operator can observe their operational cost $J(e,\esh)$. Thus, if the system operator wishes to design payments that incentivize truthful reporting $\esh = e$, we need to explore how this observable quantity $J(e,\esh)$ is affected by the decisions of the strategic sensor. 

First, let's provide an overview of how $e$ and $\esh$ influence our system. The LQR gain $K_k$ does not depend on $\esh$ at all. The estimation gain $L_k(\esh)$ depends on $\esh$ but not $e$. Neither the observer or controller gains are affected by $e$, since the system operator does not know $e$. However, the realized values of the noise $y$ depend on $e$, so the resulting state estimate $\hat{x}$ depends on $e$ and $\esh$. Consequently, $u = K_k \hat{x}$ depends on both terms, and so does $J$.

The key insight of this section is that $J$ admits a decomposition into terms that depend on $e$ and terms that depend on $\esh$. This is given in Equation~\eqref{eq:actual_J}. This allows us to design the payments in Section~\ref{sec:payment} in a fashion that incentivizes truthful reporting.

\subsection{A decomposition of $\mb{E}[\Jes]$}

We can decompose the expected cost to separate its dependence on $e$ and $\esh$. 

Looking at the closed-loop system, with observer and controller gains parameterized by $\esh$, let us define the augmented state and dynamics $\Bar{x}, \Bar{A}(\esh), \Bar{B}(\esh)$, and $\Bar{v}$:
\begin{equation}
\label{eq:bar_defs1}
\underbrace{ 
\begin{bmatrix}
x_{k+1}\\
\hat{x}_{k+1}
\end{bmatrix}
}_{\Bar{x}_{k+1}} = 
\underbrace{ 
\begin{bmatrix}
A & BK_k \\ 
L_{k+1}(\esh)CA & A+BK_k - L_{k+1}(\esh)CA
\end{bmatrix}}_{\Bar{A}_k(\esh)}
\underbrace{
\begin{bmatrix}
x_{k}\\\hat{x}_{k}
\end{bmatrix}}_{\Bar{x}_{k}} 
\end{equation}
 \[
+ \underbrace{\begin{bmatrix}
I & 0\\ L_{k+1}(\esh) C & L_{k+1}
\end{bmatrix}}_{\Bar{B}_k(\esh)} \underbrace{\begin{bmatrix}
w_{k}\\v_{k+1}
\end{bmatrix}}_{\Bar{v}_k(e)}
\]
\begin{equation}
\label{eq:sigma_v}
\Sigma_{\Bar v}(e) =
\begin{bmatrix}
\Sigma_w & 0\\ 0 &\Sigma_v(e)
\end{bmatrix} =
\begin{bmatrix}
\Sigma_w & 0 & 0 \\
0 & \Sigma_{v_r} & 0 \\
0 & 0 & \sigma^2(e) I
\end{bmatrix}
\end{equation}
Furthermore, note that we can recursively apply Equation~\eqref{eq:bar_defs1} to yield:
\begin{equation}
\label{eq:augment_system}
\begin{aligned} 
\Bar{x}_{k+1} & = \Bar{A}_k(\esh) \Bar{x}_{k} +  \Bar{B}_k(\esh) \Bar{v}_k \\ & = \underbrace{\Bar{A}_{k} \cdots \Bar{A}_1 \Bar{A}_0}_{\mathbf{\Bar{A}}_{k}(\esh)}   \Bar{x}_{0} + \underbrace{\begin{bmatrix}
(\Bar{A}_{k}\cdots\Bar{A}_{1}\Bar{B}_0)^T \\  (\Bar{A}_{k}\cdots\Bar{A}_{2}\Bar{B}_1)^T \\  \vdots \\  (\Bar{B}_{k})^T
\end{bmatrix}^T}_{\mathbf{\Bar{B}}_{k}(\esh)} \underbrace{\begin{bmatrix}
\Bar{v}_0 \\  \Bar{v}_1 \\  \vdots \\ \Bar{v}_k
\end{bmatrix}}_{\mathbf{\Bar{V}}_{k}(e)}
\end{aligned}
\end{equation}
Finally, let's separate out the appropriate cost terms by their dependence on $e$ and $\esh$. Let's define the following:
\begin{equation}
\label{eq:bold_q}
\mf{\bar{Q}_k} = 
\begin{bmatrix}
Q & 0\\ 0 & K^T_k R K_k
\end{bmatrix}
\end{equation}
\begin{equation}
\label{eq:bold_sigk}
\mathbf{\bar{\Sigma}}_k(e) = \mb{E}[\mathbf{\bar{V}_k}\mathbf{\bar{V}_k}^\T] = \begin{bmatrix}
\mathbf{\Sigma}_{\Bar{v}(e)} & & \\ & \ddots & \\ & & \mathbf{\Sigma}_{\Bar{v}(e)} \end{bmatrix}
\end{equation}

From Equation~\eqref{eq:sigma_v}, we can extract the $\sigma^2(e)$ from $\mathbf{\bar{\Sigma}}_k(e)$ and define constant matrices $\mf{\bar{\Sigma}}_{k,1}$ and $\mathbf{\bar{\Sigma}}_{k,2}$ such that:
\begin{equation}
\label{eq:sigk1_defs}
\mf{\bar{\Sigma}}_k(e) = \mf{\bar{\Sigma}}_{k,1}  +  \sigma^2(e) \mf{\bar{\Sigma}}_{k,2}
\end{equation}

Similarly, let's look at the initial covariance of our augmented state:
\[
\begin{aligned} & \qquad 
\mb{E}(\Bar{x}_{0}\Bar{x}^T_{0})  \\ & =
\begin{bmatrix}
\Sigma_{x_0} 				& \Sigma_{x_0} C^T L^T_0(\esh) \\ 
L_0(\esh) C \Sigma_{x_0} 	& L_0(\esh) (C \Sigma_{x_0}C^T + \begin{bmatrix}\Sigma_{v_r} & 0 \\ 0 & \Sigma_{v_s} (e)\end{bmatrix}) L^T_0(\esh)
\end{bmatrix}
\end{aligned}
\]
Accoording to \eqref{eq:effort-covariance}, $\mb{E}(\Bar{x}_{0}\Bar{x}^T_{0}) $ can be written as: 
\begin{equation}
\label{eq:var_augment_system}
\mb{E}(\Bar{x}_{0}\Bar{x}^T_{0}) = \mf{\Sigma}_1(\esh) + \sigma^2(e) \mf{\Sigma}_2(\esh)
\end{equation}

Therefore, noting that $K_N = 0$ and the independence of noise across time, we can rewrite $\Jes$ as:
\[
\begin{aligned}
\mb{E}[\Jes] & = \mb{E} \left[\sum_{k=0}^{N}\Bar{x}^T \mf{\Bar{Q}_k} \Bar{x} \right] \\ 
& = \sum_{k=1}^{N} \Bigg( \operatorname{tr} \Big(\mf{\bar{Q}}_k( \mathbf{\Bar{A}}_{k-1}(\esh) \mb{E}(\Bar{x}_{0}\Bar{x}^T_{0})\mathbf{\Bar{A}}^T_{k-1}(\esh)\\ & +  \mathbf{\Bar{B}}_{k-1}(\esh) \mf{\bar{\Sigma}}_k(e)\mf{\Bar{B}}^T_{k-1}(\esh)\Big) \Bigg) \\ 
& + \operatorname{tr} ( \mf{\bar{Q}}_0 \mb{E}(\Bar{x}_{0}\Bar{x}^T_{0})) 
\end{aligned}
\]

Along with $\mathbf{\bar{\Sigma}}_k(e)$ defined in (\ref{eq:sigk1_defs}), we can represent $\mb{E}[\Jes]$ in a concise form:
\begin{equation}
\label{eq:actual_J}
\mb{E}[\Jes]  = f_1(\esh) + \sigma^2(e)f_2(\esh)
\end{equation}
Where $f_1(\esh)$ and $f_2(\esh)$ are non-negative functions that are sum of traces of multiplications of positive definite matrices:
\begin{equation}
\label{eq:f_i}
\begin{aligned}
f_i(\esh) = &
= \sum_{k=1}^{N} \Bigg( \operatorname{tr} \Big(\mf{\bar{Q}}_k( \mathbf{\Bar{A}}_{k-1}(\esh) \mf{\Sigma}_i(\esh) \mathbf{\Bar{A}}^T_{k-1}(\esh)\\ & +  \mathbf{\Bar{B}}_{k-1}(\esh) \mf{\bar{\Sigma}}_{k,i} \mf{\Bar{B}}^T_{k-1}(\esh)\Big) \Bigg) \\ 
& + \operatorname{tr} ( \mf{\bar{Q}}_0  \mf{\Sigma}_i(\esh))  \qquad \text{ for } i = 1,2
\end{aligned}
\end{equation}

\subsection{Properties of $\mb{E}[\Jes]$}

Given the decomposition in Equation~\eqref{eq:actual_J}, we can derive properties of $\mb{E}[\Jes]$.

First, we fix the reported effort $\esh$ and view this as a function of the true effort $e$.

\begin{proposition}
For a fixed $\esh$, $e \mapsto \mb{E}[\Jes]$ is convex and strictly decreasing.
\end{proposition}
\begin{proof}
Noting the decomposition in Equation~\eqref{eq:actual_J}, for a fixed $\esh$, we can see that this mapping is the composition of a convex function with an affine function. The fact that this function is strictly decreasing follows from the fact $\sigma^2(\cdot)$ is strictly decreasing.
\end{proof}

Next, we fix the true effort $e$ and look at the reported effort $\esh$. Note that the LQG controller specified in Section~\ref{sec:sys_op} is optimal under truthful reporting, since, in this case, our problem reduces to the classical LQG control problem.

Based on empirical evidence, we conjecture that it seems that $f_1$ is generally decreasing and convex whereas $f_2$ is generally increasing and concave. (Varying the system parameters yielded this result consistently.) 

Consider the following 2-D example with time horizon $N = 300$:
\[
A =
\begin{bmatrix}
0.7 & 0 \\
0.7 & 0.7 
\end{bmatrix}
\quad
B = 
\begin{bmatrix}
1 \\ 0 
\end{bmatrix}
\quad
C^r = 
\begin{bmatrix}
1 & 0
\end{bmatrix} 
\]
\[
C^s = 
\begin{bmatrix}
0 & 1
\end{bmatrix}
\quad
\Sigma_{v_s} = 1 \quad R = 1
\]
\[
\Sigma_{x_0} = \Sigma_{w} = \begin{bmatrix}
1 & 0 \\ 0 & 1
\end{bmatrix}
\quad
Q = 
\begin{bmatrix}
1 & 0 \\ 0 & 1
\end{bmatrix}
\quad
\Sigma_{v_h} = \sigma^2(e) = \frac{1}{e}
\]
The plots of $f_1$ and $f_2$ for this system are shown in Figure~\ref{fig:f1_and_f2}.

\begin{figure}[!h]
\centering
\includegraphics[width=0.5\textwidth]{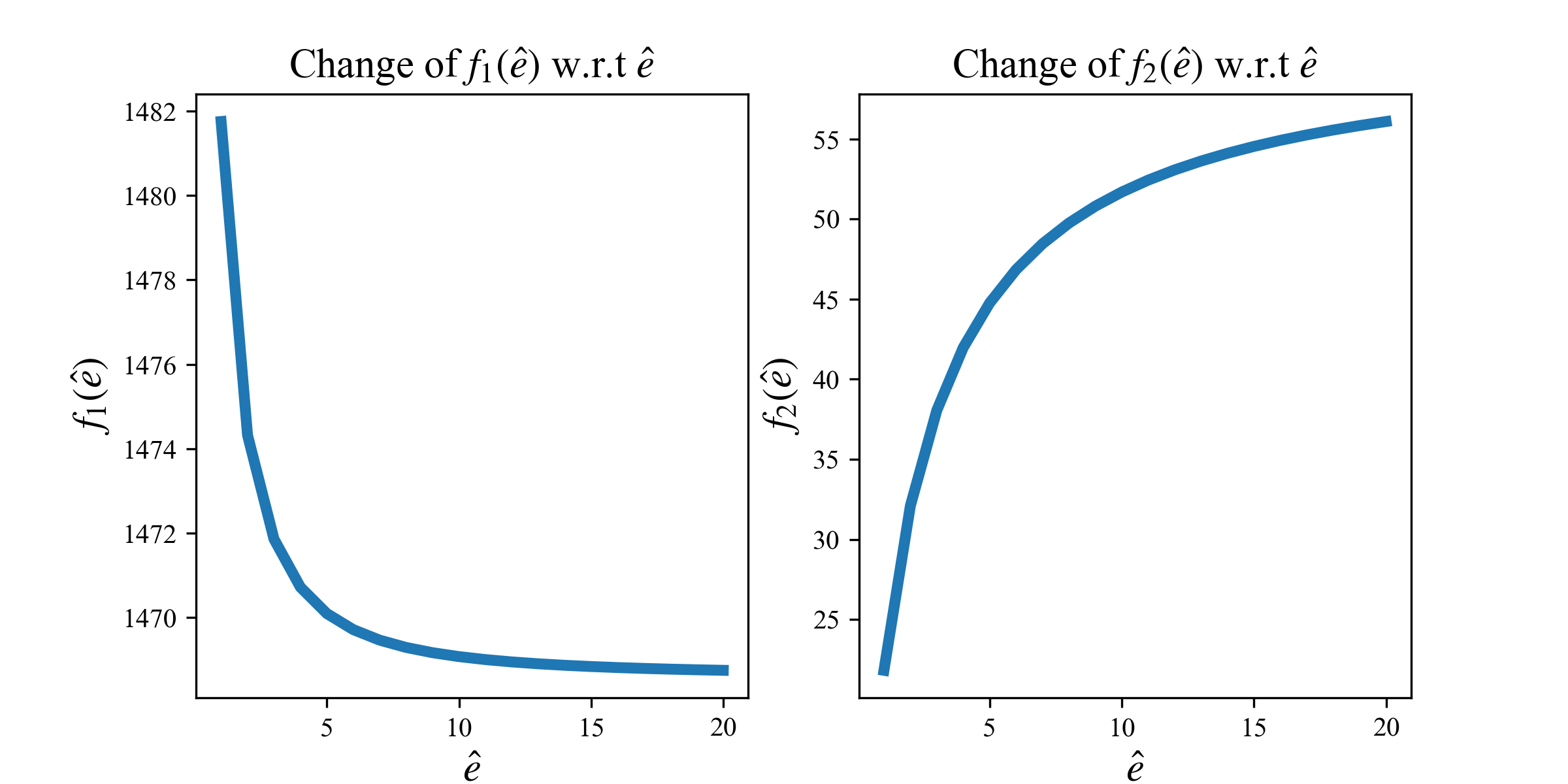}
\caption{\label{fig:f1_and_f2} $f_1(\esh)$ and $f_2(\esh)$ as a function of $\esh$ for the 2-D example described in Section~\ref{sec:cost}.}
\end{figure}

If this conjecture is true, the decomposition implies that $(e, \esh) \mapsto \mb{E}[\Jes]$ is neither convex nor concave, which is shown in Proposition~\ref{prop:convexity_of_J}. 
\begin{proposition}
\label{prop:convexity_of_J}
If $f_1$ is convex and $f_2$ is concave (or vice versa), then $(e,\esh) \mapsto \mb{E}[\Jes]$ is neither convex nor concave.
\end{proposition}
\begin{proof}
From Equation~\eqref{eq:actual_J}, $\mb{E}[\Jes]$ will be dominated by $f_2(\esh)$ with small $e$, and will be dominated by $f_1(\esh)$ for a large $e$.
\end{proof}

At a minimum, the 2-D example shows that the mapping $(e,\esh) \mapsto \mb{E}[\Jes]$ is not concave nor convex, and this property will hold in general.

Additionally, let $J^*$ denote the expected cost under truthful reporting:
\begin{equation}
\label{eq:jstar}
J^*(\esh) = \mb{E}[J(\esh,\esh)]
\end{equation}
In this case, the problem reduces to a classical control problem and we can see $J^*(\esh)$ decreases as $\esh$ increases.

\begin{proposition}
\label{prop:J_decrease}
The optimal cost $J^*(\esh)$ decreases as $\esh$ increases.
\end{proposition}
\begin{proof} 
Note that:
\begin{equation}
\label{eq:jopt}
\begin{aligned} 
\Jesh  = & \sum_{k=0}^{N-1} \operatorname{Tr}\left(Q \Sigma_{k}(\esh)\right)+ Q_N \Sigma_{N} \\ & +\sum_{k=0}^{N} \operatorname{Tr} P_{k}\left(\Sigma_{k | k-1}(\esh)-\Sigma_{k}(\esh)\right)
\end{aligned}
\end{equation}
From Equation~\eqref{eq:jopt}, it suffices to show that 
\[
\Sigma_k(\esh) \succ \Sigma_k(\esh') \qquad \text{if} \quad \esh' > \esh
\]
For the state $x_k$ and history of observations $y_0, \dots, y_k$, these random vectors are jointly Gaussian:
\[
\left[\begin{array}{c}{x_{k}} \\ {y_{0}} \\ {\vdots} \\ {y_{k}}\end{array}\right] \quad \sim \quad \mathcal{N}\left(\mathbf{0}, \mf{\Sigma}\right) \quad  \mf{\Sigma} = \begin{bmatrix}
\mf{\Sigma}_{xx} & \mf{\Sigma}_{xy} \\ \mf{\Sigma}_{yx} & \mf{\Sigma}_{yy}
\end{bmatrix}
\]
Note that $\mf{\Sigma}_{xy}$ and $\mf{\Sigma}_{yx}$ do not depend on $\esh$ because the input $u$ is known and therefore $\mf{E}[xy]$ does not depends on $\esh$. That is, only $\mf{\Sigma}_{yy}$ depends on the measurement noise covariance, and it can be written as a form of $\Sigma_{yy} = M + \sigma^2(\esh)D$, where $M$ is a  positive definite matrix and $D$ is a diagonal matrix. Then the covariance of the conditional distribution $x_k | y_0,y_1, \dots, y_k$ is $\mf{\Sigma}_{xx} - \mf{\Sigma}_{xy} {\mf{\Sigma}}_{yy}^{-1} \mf{\Sigma}_{yx}$. 
Therefore $\sigma^2(\esh) > \sigma^2(\esh')$ implies $J^{\star}(\esh)> J^{\star}(\esh')$.
\end{proof}

We plot the functions $\esh \mapsto \mb{E}[\Jes]$ and $J^*$ in Figure~\ref{fig:EJ_and_Jstar}. We can see that for a fixed $e$, the shape of $\esh \mapsto \mb{E}[\Jes]$ changes with different values $e$. For small values of $e$, the function looks concave for the range of domain values plotted, and for larger values of $e$, the function looks convex. 
\begin{figure}[!htbp]
\centering
\includegraphics[width=0.45\textwidth]{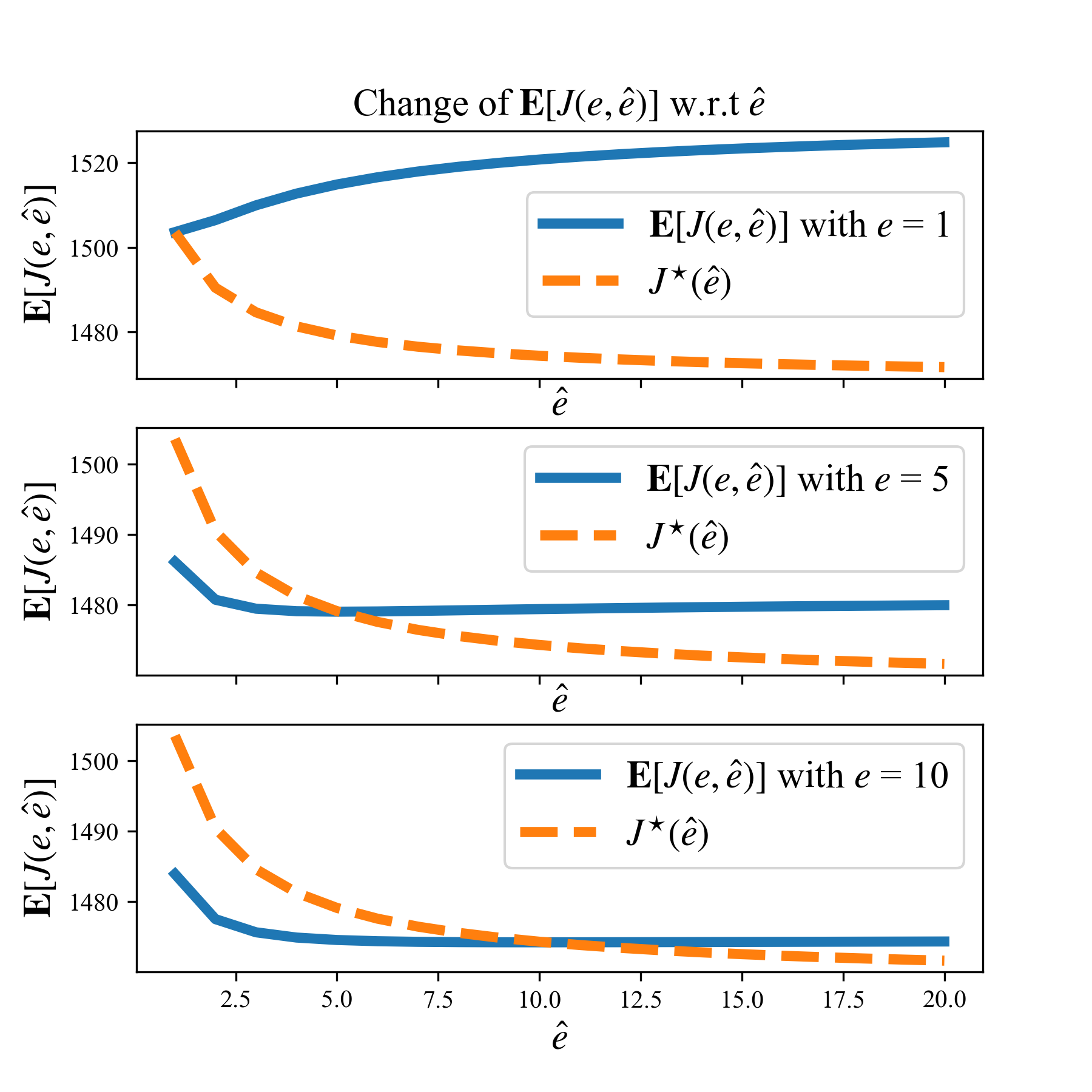}
\caption{\label{fig:EJ_and_Jstar} The functions $\esh \mapsto \mb{E}[\Jes]$ and $J^*(\esh)$ as a function of $\esh$, for different values of $e$. (Note $J^*(\esh)$ does not depend on $e$.)}
\end{figure}

\section{Exploration of Payment Functions}
\label{sec:payment}


The main goal of this work is to study how a system operator can design payment functions $p$ to incentivize a strategic data source 
to truthfully report their effort $\esh = e$. In this section, we construct payment functions $p$ which are based on the information available to the system operator. 

Note that $J^*(e)$ is simply the expected cost of the classical LQG controller, and this can be calculated a priori, according to \eqref{eq:jstar}. Additionally, the system operator can observe $\Jes$ based on the actual behavior of the underlying dynamical system.

Let's construct the payment function now.

First, let's discuss truthful reporting. That is, for a fixed $e$, we wish for the strategic sensor's optimal $\esh$ to be $\esh = e$. Note that $\mb{E}[\Jes] = \Jesh$ when $\esh = e$, and, by the decomposition in Equation~\eqref{eq:actual_J} and the properties of $\sigma^2(\cdot)$, $\mb{E}[\Jes] \neq \Jesh$ when $\esh \neq e$. Thus, the payments should penalize deviations between $\mb{E}[\Jes]$ and $\Jesh$. Since we only have access to a sample $\Jes$ rather than the expected value $\mb{E}[\Jes]$, a common way to do this is through a quadratic penalty, which can decompose into a bias term and a direct penalty on the quantity of interest.
\[
(\Jes - \Jesh)^2 = (\Jes - \mb{E}[\Jes])^2
\]
\[
+ (\mb{E}[\Jes] - \Jesh)^2
\]
However, payment with only quadratic penalty will make sensors exert small effort, although they will truthfully report it. Therefore, we would like the sensor to internalize the operational cost: the lower the realized value $\Jes$, the higher the payment should be. Thus, we define the payment:
\begin{equation}
\label{eq:simple_p}
p_0(\Jesh,J(e, \esh)) = a - b_J (J(e, \esh) - \Jesh)^2 - b_e J(e, \esh)
\end{equation}
Where $a$, $b_J$, and $b_e$ are non-negative constants. The term with $b_J$ is to incentivize truthful reporting, and the term with $b_e$ is to incentivize sensors to choose a higher value of $e$. This form of payment function has been commonly used in many related literature, i.e. \cite{D.G_2016, T.W_2019, cai2014optimum}. In our model, we assumed the strategic sensor acts ex-ante and is risk-neutral, so we can see the expected value of this payment is:
\[
\begin{aligned} &
\mb{E}\left[p_0(\Jesh,J(e, \esh))  \right]  \\ & \quad =  a - b_J \mb{E}\left[(J(e, \esh) - \Jesh)^2 \right] - b_e\mb{E}\left[J(e, \esh)\right] \\ & \quad =  a - b_J\left(\mf{Var}[J(e, \esh)] + (\sigma^2(e) - \sigma^2(\esh))^2 f^2_2(\esh)\right) \\ & \qquad - b_e(f_1(\esh) + \sigma^2(e) f_2(\esh))
\end{aligned}
\]
However, the optimal point of the sensor's utility, $\mf{E}\left[p_0(\Jesh,J(e, \esh)) \right] - e$, depends on the properties of $f_1(\esh)$ and $f_2(\esh)$. As shown in Section~\ref{sec:cost} in Equation~\eqref{eq:f_i}, we can calculate $f_1$ and $f_2$ given the system parameters. However, none of our previous results would imply that the sensor's expected utility would be minimized by $\esh = e$. Furthermore, the term $\mf{Var}[J(e, \esh)]$ also depends on $e$ and $\esh$, and does not disappear as other non-dynamic applications. Therefore, we can not guarantee this payment incentivizes the strategic sensor to truthfully report $\esh = e$. This leads to our next proposition.

\begin{proposition}
In general, the payment contract $p_0$, given in Equation~\eqref{eq:simple_p}, does not incentivize truthful reporting.
\end{proposition}
\begin{proof}
Follows from the previous discussion.
\end{proof}

To resolve this problem, we wish to make sure that $\esh = e$ is the optimal point, even as $f_1(\esh)$, $f_2(\esh)$ and $\mf{Var}[J(e, \esh)]$ change. Because $f_1(\esh)$ and $f_2(\esh)$ can be calculated ex-ante, we modify the payment in Equation~\eqref{eq:simple_p} by making the coefficient $b_J$ depend on $\esh$, and also introducing the $f_2(\esh)$ term. 
Let $b_J$ be any non-negative, convex, twice continuously differentiable function. 
Then, we can define the new payment $p(\Jesh,J(e, \esh))$ as:
\begin{equation}
\label{eq:def_p}
\begin{aligned} &
p(\Jesh,J(e, \esh))\\ & =a - b_{J}(\esh)\left(\frac{J(e, \esh) - J^*(\esh)}{f_2(\esh)}\right)^2 \\ & \quad - b_{e}\left(\frac{J(e, \esh) - J^*(\esh)}{f_2(\esh)} + \sigma^2(\esh)\right)
\end{aligned}
\end{equation}
Taking the expectation of this payment function, and noting our decomposition in Equation~\eqref{eq:actual_J}:
\begin{equation}
\label{eq:expected_p}
\begin{aligned} &
\mb{E}\left[p(\Jesh,J(e, \esh))\right]\\ &  =a - b_{J}(\esh)\mb{E}\left(\frac{(J(e, \esh) - \Jesh}{f_2(\esh)}\right)^2 \\  
& \quad - b_{e}\left(\frac{\mb{E}[J(e, \esh)] - \Jesh}{f_2(\esh)}  + \sigma^2(\esh)\right) \\
& = \underbrace{- \frac{b_{J}(\esh)}{f_2^2(\esh)}\mf{Var}[J(e, \esh)]}_{p_{var}} + \ps
\end{aligned} 
\end{equation}
Where:
\begin{equation}
\label{eq:pstar}
\ps = a - b_{J}(\esh)(\sigma^2(e) - \sigma^2(\esh))^2- b_{e}\sigma^2(e)
\end{equation}
By construction, the $f_1(\esh)$ terms cancel out and it does not show up in \eqref{eq:expected_p}.
Additionally, $b_{J}(\esh)$ can be chosen carefully to compensate for the $f_2(\esh)$ and $\mf{Var}[J(e, \esh)]$ terms. Therefore, by tuning the parameters of this payment function, we can set $\esh = e$ as a local maximum for the sensor's utility. We do so by matching terms based on the first and second derivatives to recover sufficient conditions for optimality. However, noting our reliance on differentiation, we are only able to guarantee local optimality.

We first show that  $\esh = e$ is a local optimal point for the term $\ps$ with respect to $\esh$.
\begin{lemma}
\label{lem:local_max}
Fix an effort level $e$. 
We have that $\esh = e$ is always a local maximum for $\esh \mapsto \ps$. 
\end{lemma}
\begin{proof}
At $\esh = e$, the term $\sigma^2(e) - \sigma^2(\esh) = 0$. Thus, the first derivative is 0, and all but one term in the second partial derivative of $\ps$ disappears. Noting that $\sigma^2$ is convex, we have: 
\[
\begin{aligned}
\left.
\frac{\partial^2 p^{\star} }{\partial \esh^2} 
\right|_{\esh = e} &= -2 b_J(\esh)\left(\frac{d\sigma^2(\esh)}{d\esh}\right)^2 < 0
\end{aligned}
\]
Thus $\esh = e$ is a local maximum.
\end{proof}
Now we consider the $p_{var}$ term in (\ref{eq:expected_p}).

First, note that $\mf{Var}[J(e, \esh)]$ can be written as:
\[
\begin{aligned} &
\mf{Var}[J(e, \esh)] \\ & = 2\sum_{k=0}^{N}\operatorname{tr} (
\mf{\bar{Q}_k}
\mb{E}[\bar{x}_k \bar{x}^T_k] 
\mf{\bar{Q}_k}
\mb{E}[\bar{x}_k \bar{x}^T_k]) \\ & + \sum_{k=0}^{N} \sum_{j \neq k}^{N} \mf{Cov}\left( \bar{x}^T_k
\mf{\bar{Q}_k}
\bar{x}_k, \bar{x}^T_j
\mf{\bar{Q}_k}
\bar{x}_j\right)
\end{aligned}
\]
Note that this is the variance of a quadratic cost, so it contains fourth-order terms of the Gaussian variables. This keeps us from the simpler analysis that applies to first and second moments.

Additionally, $\mf{Var}[J(e, \esh)]$ is a function of $e$ and $\esh$ that is twice-differentiable. Because we are considering truthful reporting (i.e. $\esh = e$ should be optimal), we only care about the local properties around $\esh = e$. $\mf{Var}[J(e, \esh)]$ and its derivatives can be calculated ex-ante at point $\esh = e$. Therefore, with knowledge of $\mf{Var}[J(e, \esh)]$, we can force $\esh = e$ to be a local maximum by carefully choosing $b_J(\esh)$.
\begin{theorem}
\label{thm:variance}
Fix an effort level $e$. 
If the following conditions holds at $\esh = e$, then $\esh = e$ is the local maximum for $\mb{E}[p]$. 
\begin{enumerate}
    \item $\left. \frac{d^2b_J(\esh)}{d \esh^2}\right|_{\esh = e}$ is sufficiently large.
    \item The following equality holds at $\esh = e$:
\begin{equation}
\label{eq:conditions_Ep}
\begin{aligned}& 
\frac{db_J(\esh)}{d \hat{e}} \mf{Var}[J(e, \esh)] {f_2}^{-2}(\esh) \\ & - 2b_J(\esh) \mf{Var}[J(e, \esh)] {f_2}^{-3}(\esh) \frac{df_2(\esh)}{d \esh} \\ &
+ b_J(\esh) 
\frac{\partial\mf{Var}[J(e, \esh)]}{\partial\esh} {f_2}^{-2}(\esh) = 0
\end{aligned}
\end{equation}
\end{enumerate}
\end{theorem}
\begin{proof}
We have shown that $\esh = e$ is a local maximum for $\ps$ in Lemma~\ref{lem:local_max}. Differentiating \eqref{eq:expected_p},  $\esh = e$ is a local maximum for $\mf{E}\left[p(J(e, \esh), \Jesh)\right]$ if it is also a local maximum for $p_{var}$.
$$
\frac{\partial^2p_{var}}{\partial \esh^2} = - \frac{d^2b_J(\esh)}{d\hat{e}^2} \mf{Var}[J(e, \esh)] {f_2}^{-2}(\esh) + \text{other terms}
$$
Therefore, $\frac{d^2p}{d\hat{e}^2_s}$ is negative if $\frac{d^2b_J(\esh)}{d\hat{e}^2_s}$ is sufficiently large at $\esh = e$.
The Equation~\eqref{eq:conditions_Ep} is actually the first-order optimality condition: $\frac{dp}{d\hat{e}_s} = 0$. 
\end{proof}
Therefore, by choosing the parameters of $p$ appropriately, payment functions of the form in Equation~\eqref{eq:def_p} can incentivize truthful reporting on the part of strategic sensors. When the strategic sensor truthful reports, the LQG controller described in Section~\ref{sec:sys_op} is optimal, in the sense that it minimizes the expected cost, given the noise covariances.

\section{Conclusion}
\label{sec:conclusion}


In this paper we analyzed the effort-averse strategic behavior of a single type of strategic sensors on the expected quadratic cost of a linear system with Gaussian noise. We studied how the expected cost $\mb{E}[J]$ varies with both the true effort $e$, which determines the measurement noise, and the reported effort $\esh$, which influences the gains of the observer and controller. Surprisingly, we can decompose $\mb{E}[J]$ into terms that depend soley on $e$ and $\esh$. This allows us to learn how the expected cost varies with both of these parameters, and, additionally, allows us to define payment contracts that can incentivize truthful reporting in expectation. We show that payment contracts typically used in static settings fail to incentivize truthful reporting in general, and then modify the mechanism using values that can be calculated a priori, which incorporate the closed-loop effect of reported data into the payment itself.

We view this work as the first step into understanding how to incentivize strategic data sources in dynamic settings. In particular, we show that the mechanisms designed in the static setting will often not work in dynamical settings due to the closed-loop behavior of the system. Additionally, many parameters from the closed-loop system likely must be incorporated in order to achieve desirable properties, such as truthful reporting. In future work, we hope to consider how the system operator behaves when the payments are costly: how does the system operator trade-off the cost of issuing payments with the system's operational costs?




\bibliographystyle{IEEEtran}
\bibliography{./Yuelin-refs}


\end{document}